\documentclass[10pt,letterpaper]{article}
\usepackage[letterpaper]{geometry}
\usepackage{hicss}
\usepackage{times}
\usepackage[none]{hyphenat}
\usepackage{url}
\usepackage{latexsym}
\usepackage{minted}
\usepackage{indentfirst}
\usepackage{graphicx}
\graphicspath{{images/}}

\usepackage{amsmath,amssymb,amsthm,epsfig,color,subfigure,empheq,graphicx,graphics,balance}
\usepackage{enumerate,url,algorithm,algorithmic,wasysym,epstopdf,enumitem}
\usepackage{xcolor}
\usepackage{amsfonts}
\usepackage{adjustbox}
\usepackage{multirow}
\usepackage{accents}
\usepackage{placeins}

\DeclareMathOperator{\real}{Re}
\DeclareMathOperator{\imag}{Im}

\newtheorem{lemma}{Lemma}

\newcommand \bzero{\mathbf{0}}
\newcommand \bone{\mathbf{1}}
\newcommand \ba{\mathbf{a}}
\newcommand \bb{\mathbf{b}}
\newcommand \bc{\mathbf{c}}

\newcommand \be{\mathbf{e}}

\newcommand \bs{\mathbf{s}}

\newcommand \bv{\mathbf{v}}
\newcommand \bw{\mathbf{w}}
\newcommand \bx{\mathbf{x}}

\newcommand \bz{\mathbf{z}}
\newcommand \bA{\mathbf{A}}

\newcommand \bE{\mathbf{E}}

\newcommand \bH{\mathbf{H}}
\newcommand \bI{\mathbf{I}}
\newcommand \bJ{\mathbf{J}}
\newcommand \bK{\mathbf{K}}

\newcommand \bM{\mathbf{M}}

\newcommand \bP{\mathbf{P}}
\newcommand \bQ{\mathbf{Q}}

\newcommand \bU{\mathbf{U}}
\newcommand \bV{\mathbf{V}}
\newcommand \bW{\mathbf{W}}

\newcommand \bY{\mathbf{Y}}

\newcommand \balpha{\boldsymbol{\alpha}}
\newcommand \bbeta{\boldsymbol{\beta}}

\newcommand \bdelta{\boldsymbol{\delta}}

\newcommand \bsigma{\boldsymbol{\sigma}}

\newcommand \bDelta{\mathbf{\Delta}}

\newcommand \bLambda{\mathbf{\Lambda}}

\newcommand \mcC{\mathcal{C}}

\newcommand \mcG{\mathcal{G}}

\newcommand \mcL{\mathcal{L}}

\newcommand \mcN{\mathcal{N}}
\newcommand \mcO{\mathcal{O}}
\newcommand \mcP{\mathcal{P}}

\newcommand \mcR{\mathcal{R}}

\newcommand \mcV{\mathcal{V}}

\newcommand \tbv{\tilde{\mathbf{v}}}

\newcommand \cbdelta{\check{\boldsymbol{\delta}}}

\newcommand \bbv{\bar{\mathbf{v}}}

\newcommand \bbY{\bar{\mathbf{Y}}}

\newcommand \bbDelta{\bar{\mathbf{\Delta}}}

\title{AC Power Flow Contingency Analysis Using a Single Deep Neural Network}

\author{Md Obaidur Rahman, Junjie Qin, and Vassilis Kekatos\\
  Elmore Family School of Electrical and Computer Engineering \\
Purdue University, West Lafayette, IN 47906\\
 {\underline{ rahma160, jq, kekatos@purdue.edu}}} 

\begin{document}
\maketitle

\begin{abstract}
Contingency analysis using the AC power flow (AC-PF) model is a critical tool for accurate grid security assessment, but its computational burden increases with the number of operating scenarios and outage configurations to evaluate. Recent ML-based approaches typically require outage-specific training data, leading to offline training costs that scale with the number of contingencies. This work proposes a framework that reuses a single ML model trained solely on basecase AC-PF data to estimate post-contingency operating states under arbitrary single-line outages. The proposed approach formulates post-contingency state prediction as a fixed-point iteration. If the ML model is a deep neural network (DNN), we derive sufficient conditions that guarantee convergence and develop semidefinite programming (SDP) formulations to certify these conditions for a given DNN. Numerical tests on the IEEE 118-bus system demonstrate that the proposed SDP formulations are tight, that the certified conditions hold for all tested contingencies, and that the resulting method produces accurate post-contingency state estimates within only a few iterations.
\end{abstract}


\allowdisplaybreaks

\section{Introduction}
Contingency analysis is a fundamental tool for ensuring the secure and reliable operation of power systems and for evaluating future grid expansion plans. The need for scalable contingency analysis is growing as rapid load growth, renewable integration, transportation electrification, and large data centers expand the range of operating conditions that must be assessed. Moreover, corrective switching and topology optimization may lead to more dynamic network configurations. Since contingency analysis relies on repeated power-flow computations across many operating scenarios and outage configurations, its computational burden can become substantial. In this context, we propose a novel framework for predicting post-contingency power system states using a single machine-learning model that can be reused across multiple outage scenarios.

To alleviate the computational burden of AC contingency analysis, operators often rely on sensitivity-based approximations, such as power-transfer and line-outage distribution factors (PTDFs and LODFs)~\cite{sauer1981formulation,wood2013power}. Numerous works have improved the efficiency of these methods through direct LODF calculations, graph-theoretic formulations, and extensions to multiple-line outages~\cite{guo2009direct,castelli2024improved,guler2007generalized}. However, they cannot accurately capture voltage violations, reactive power redistribution, network losses, or apparent power limit violations. To improve accuracy, several works have developed AC sensitivity factors and linearized AC-PF models~\cite{qu2015uncertainty,dhople2015linear,coffrin2014linear}, yet their accuracy deteriorates under topology changes or major shifts in the operating point, limiting their applicability to screening and approximate contingency analysis.

Another line of work seeks to accelerate exact AC contingency analysis. Warm-started Newton--Raphson methods initialize the post-contingency solve using the basecase voltage profile, thereby reducing the number of iterations required for convergence \cite{yu2024efficient}. Other approaches exploit the fact that a single-line outage induces a low-rank modification of the AC-PF Jacobian. Using the matrix inversion lemma, these methods update the inverse of the basecase Jacobian using outage information rather than recomputing the solution from scratch~\cite{yeung2017amps,lo2004newton}. Holomorphic embedding methods employ analytic continuation to determine whether a post-contingency solution exists and to avoid Newton--Raphson divergence near stressed operating points \cite{yao2021contingency}, but their complexity can still be a concern.

These challenges have motivated the development of ML-based surrogates for power-system computations. DNNs and related ML models have been successfully applied to a variety of power-system tasks, such as the optimal power flow, emergency control, and security assessment~\cite{zamzam2020learning,pan2021deepopf, singh2022learning,zhou2022scalable,gupta2023optimal,chen2020learning,kody2023learning}. For the AC-PF and contingency analysis, physics-informed neural networks incorporate AC-PF residuals, Kirchhoff's law constraints, or network topology information during training to improve voltage predictions \cite{hu2021physics,eeckhout2024improved}. Other works train ML models to predict post-contingency voltages, line loadings, and security indices directly from pre-contingency operating conditions \cite{du2019achieving}. To account for topology changes, graph neural networks explicitly encode buses and transmission lines for contingency prediction and screening \cite{nakiganda2025graph}, while input-convex DNNs have been proposed for reliable $N-k$ contingency screening with guarantees against false negatives \cite{christianson2025fast}.

Despite recent progress, existing ML-based approaches for contingency analysis often require outage-specific training data, increasing offline costs for data generation and model training as the number of contingencies grows. To overcome this limitation, we propose an ML-based fixed-point framework for AC contingency analysis. The contributions are threefold: \emph{c1)} We derive a formulation that encodes line outages as input perturbations, enabling a single DNN trained on the basecase topology to be reused across contingencies; \emph{c2)} We establish sufficient conditions guaranteeing convergence of the iterative scheme; and \emph{c3)} We develop SDP formulations for certifying these conditions. Numerical tests on the IEEE 118-bus system demonstrate that the proposed SDPs derive tight bounds, that the certified conditions hold for all tested single-line contingencies, and that the proposed method produces accurate post-contingency state estimates within only a few iterations.

The rest of the paper is organized as follows. Section~\ref{sec:basecase_dnn} reviews the AC-PF model and basecase DNN surrogate. Section~\ref{sec:method} introduces the fixed-point framework for single-line contingencies. Section~\ref{sec:certify} derives convergence conditions, while Section~\ref{sec:sdp} develops SDP-based bounds to certify them. Section~\ref{sec:tests} evaluates the method on the IEEE 118-bus system.

\section{ML Surrogate of an AC-PF Solver}
\label{sec:basecase_dnn}
This section reviews the AC power flow (AC-PF) problem and explains how a deep neural network (DNN) can be trained to predict the solutions of an AC-PF solver under the basecase power system topology.

An AC power transmission system can be represented by a graph $\mcG=(\mcN,\mcL_0)$, where $\mcN:=\{1,\ldots,N\}$ is the set of buses and $\mcL_0$ is the set of lines or transformers. We use the subscript $0$ to indicate that this graph corresponds to the basecase power system topology. The set $\mcN$ is partitioned into the set of generators or PV buses $\mcV$; the set of load or PQ buses $\mcP$; and the reference bus. Zero-injection buses are included in $\mcP$. For each $n\in\mcN$, let $v_n=V_ne^{j\theta_n}$ and $s_n=p_n+jq_n$ denote the voltage phasor and complex power injection into bus $n$. Vectors $\bv\in\mathbb{C}^{N}$ and $\bs\in\mathbb{C}^{N}$ collect the voltage phasors and complex power injections for all buses, respectively. 

The notation $\bbv\in\mathbb{C}^{N}$ denotes the complex conjugate of $\bv$, whereas $\tbv$ collects the real/imaginary parts of $\bv$:
\begin{equation*}
\tbv:=\left[\begin{array}{l}
\real(\bv)\\
\imag(\bv)
\end{array}\right]\in\mathbb{R}^{2N}.
\end{equation*}

The notation $\ell=(m,n)$ means that line $\ell\in\mcL$ runs between buses $m$ and $n$. For each line, let $y_\ell$ denote the series admittance and let $y_\ell^s$ denote the half-line shunt admittance. The bus admittance matrix corresponding to the basecase topology can be expressed as
\begin{equation}\label{eq:Y0_construction}
\bY_0=\sum_{\ell=(m,n)\in\mcL_0}
y_\ell \ba_\ell\ba_\ell^\top+
y_\ell^s\be_m\be_m^\top+
y_\ell^s\be_n\be_n^\top,
\end{equation}
where $\be_m$ is the $m$-th column of the identity, and $\ba_\ell:=\be_m-\be_n$. Given $\bY_0$, the complex power injections across all buses satisfy the AC power flow (AC-PF) equations
\begin{equation}\label{eq:pf0}
\bs=\bv\odot\left(\bbY_0\bbv\right)
\end{equation}
where barred symbols denote complex conjugation and $\odot$ denotes entrywise multiplication.

System operators solve the AC-PF problem to compute the power system state. Given two specified quantities for each bus, the AC-PF task amounts to solving the nonlinear equations in \eqref{eq:pf0} to find $\bv$. For each PV bus $n\in\mcV$, we specify $(p_n,V_n)$. For each PQ bus $n\in\mcP$, we specify $(p_n,q_n)$. For the reference bus, we specify $V_0$ and set $\theta_0=0$. Let us stack the specified quantities in vector $\bc\in\mathbb{R}^{2N}$. When an AC-PF solution exists and is unique within the operating region of interest, we can define the mapping from the vector of specifications $\bc$ to the AC-PF solution 
\begin{equation}\label{eq:basecase_pf_map}
G:\mathbb{R}^{2N}\rightarrow\mathbb{R}^{2N},
    \qquad
    \tbv=G(\bc).
\end{equation}

To accomplish various planning and operation tasks, system operators must repeatedly solve the AC-PF problem for numerous scenarios, including diverse loading conditions and generation schedules. To alleviate the computational burden, one obvious approach would be to replace the AC-PF solver with a ML-based surrogate. One can train a deep neural network (DNN) to predict an AC-PF solution once presented with the specification vector at its input. 

One can use a standard feed-forward DNN architecture where the input $\bc$ is propagated through $K$ layers to approximate the related AC-PF solution $G(\bc)$ at the DNN output $F_\bw(\bc)$ by feeding $\bx_0=\bc$:
\begin{subequations}\label{eq:dnn}
\begin{align}
\bz_k&=\bW_k\bx_k+\bb_k,~ k=0,\ldots,K-1\label{eq:dnn:zk}\\
\bx_{k+1} &= \sigma_k\left(\bz_k\right),~ k=0,\ldots,K-1\label{eq:dnn:xk+1}\\
F_{\bw}(\bc) &= \bW_K\bx_K+\bb_K,\label{eq:out}
\end{align}    
\end{subequations}
where $\bW_k$ and $\bb_k$ is the weight matrix and bias vector of layer $k$, and $\sigma_k(\cdot)$ is the vector-valued activation function at layer $k$. Vector $\bw$ collects all the trainable DNN parameters, that is, the vectorized weight matrices and bias vectors across all layers. Mirroring the mapping of the AC-PF solver in \eqref{eq:basecase_pf_map}, let us denote the input/output mapping of the DNN as 
\begin{equation}\label{eq:dnn_map}
F_\bw:\mathbb{R}^{2N}\rightarrow\mathbb{R}^{2N},
    \qquad
    \tbv=F_\bw(\bc).
\end{equation}

To avoid the computational burden of generating a large labeled dataset through repeated AC-PF solves, the DNN is trained in an unsupervised way. Rather than solving \eqref{eq:pf0}, an AC-PF solution can be found by minimizing the squared residual error of the AC-PF equations as~\cite{MLB15,redux} 
\begin{equation}\label{eq:pf0min}
\min_{\tbv \in \mathbb{C}^N}
    \sum_{m=1}^{2N}
    \left(
    {\tbv}^\top
    \bH_m
    {\tbv}
    -
    c_m
    \right)^2, 
\end{equation}
where the symmetric matrix $\bH_m$ corresponding to specification $m$ can be derived from $\bY_0$. An AC-PF solution yields zero cost because it satisfies ${\tbv}^\top\bH_m{\tbv}=c_m$ for all $m$, and so it is a minimizer of \eqref{eq:pf0min}.

Building on \eqref{eq:pf0min}, the DNN can be trained by minimizing the AC-PF residual errors across $S$ specification scenarios $\{\bc^s\}_{s=1}^S$:
\begin{equation*}
    \frac{1}{S}
    \sum_{s=1}^{S}
    \sum_{m=1}^{2N}
    \left(F_\bw^\top(\bc^s)
    \bH_m
    F_\bw(\bc^s)
    -
    c_m^s
    \right)^2,
\end{equation*}
where $c_m^s$ is the $m$-th entry of vector $\bc^s$. In this way, the optimal DNN parameters $\bw^*$ can be found without requiring any AC-PF solutions. Once the DNN parameters have been learned, a single forward pass of the DNN for a given $\bc$ approximates the basecase AC-PF solution map
\begin{equation}\label{eq:dnn_approx_pf_map}
F(\bc):=F_{\bw*}(\bc)\simeq G(\bc).
\end{equation}
To unclutter notation, we will often drop the dependence of $F$ on the optimal DNN parameters. Given the computational effort invested in training this DNN, one might wonder whether it could also infer AC-PF solutions under line contingencies. We next explain how to incorporate this DNN into a fixed-point iteration to infer AC-PF solutions under almost every single-line contingency.

\section{Proposed Methodology}
\label{sec:method}
To ensure grid reliability, operators must evaluate the system state under any single-line contingency. In a system with $L$ lines, contingency analysis entails solving $L$ additional instances of the AC-PF problem. Each instance involves a contingent topology in which a line has been removed from the basecase topology. Under preventive contingency management, the AC-PF specifications remain unaltered between the basecase and the contingent topologies. Under corrective contingency management, however, the operator may change generator setpoints or curtail load to handle contingencies more effectively. Contingency analysis typically ignores the outage of \emph{critical lines}. A line is deemed critical if, once removed, the power system graph becomes disconnected. 

As with the basecase topology, the operator may consider using an ML model to replace AC-PF solvers, thereby alleviating the computational burden associated with AC contingency analysis. Unfortunately, the operator should train $L$ additional ML models, one for each of the $L$ contingent topologies. The offline computational overhead exacerbates further if the operator considers double-line contingencies. Leveraging the structure of the AC-PF problem across the basecase and contingent topologies, we propose an iterative scheme that uses a single DNN trained solely on the basecase topology to predict system states for any single-line contingency.

Let contingency $\ell$ be associated with the outage of line $\ell=(m,n)\in\mcL$. The buses $m$ and $n$, which are adjacent to the outaged line, will be henceforth termed \emph{terminal buses} of contingency $\ell$. Similar to \eqref{eq:pf0}, the AC-PF equations under contingency $\ell$ read as
\begin{equation}\label{eq:pfell}
\bs=\bv\odot\left(\bbY_\ell\bbv\right).
\end{equation}
A key observation is that the bus admittance matrix $\bbY_\ell$ under contingency $\ell$ relates to $\bbY_0$ as
\begin{equation}\label{eq:Yell}
\bY_\ell=\bY_0 - \bDelta_\ell,
\end{equation}
where $\bDelta_\ell:=y_\ell \ba_\ell\ba_\ell^\top+y_\ell^s\be_m\be_m^\top
+y_\ell^s\be_n\be_n^\top$. Plugging \eqref{eq:Yell} into \eqref{eq:pfell} and rearranging provides
\begin{equation}\label{eq:pfell2}
\bs+\cbdelta_\ell(\bv)=\bv\odot\left(\bbY_0\bbv\right),
\end{equation}
where we define the mapping 
\begin{equation}\label{eq:bdelta}
\cbdelta_\ell(\bv) :=\bv\odot\left(\bbDelta_\ell\bbv\right).
\end{equation}
Equation \eqref{eq:pfell2} reveals that the AC-PF equations under contingency $\ell$ can be thought of as the AC-PF equations under the basecase topology with the complex power injections properly modified. It is easy to verify that vector $\bDelta_\ell\bv$ has only two non-zero entries: the entries $m$ and $n$ related to the buses adjacent to the outaged line $\ell$. Due to the entrywise multiplication, vector $\cbdelta_\ell(\bv)$ has non-zero values only at the terminal bus locations $m$ and $n$. Therefore, we modify only the complex power injections at terminal buses for the associated topology. Although sparse, the modification vector $\cbdelta_\ell(\bv)$ depends on the unknown system state $\bv$, so the model in \eqref{eq:pfell2} seems to be of no practical use.

Nonetheless, we can leverage \eqref{eq:pfell2} to develop an iterative scheme that finds the AC-PF solution under contingency $\ell$ using the basecase AC-PF solver:
\begin{equation}\label{eq:pfell2t}
\bs+\cbdelta_\ell(\bv_t)=\bv_{t+1}\odot\left(\bbY_0\bbv_{t+1}\right).
\end{equation}
According to \eqref{eq:pfell2}, each iteration $t$ involves two steps. The first step computes the modified complex power injections on the left-hand side (LHS) of \eqref{eq:pfell2t} using the most recent state estimate $\bv_t$. In the second step, we find $\bv_{t+1}$ as the solution of the basecase AC-PF solver fed with the modified AC-PF specifications. If the iterative scheme converges, its equilibrium coincides with the AC-PF solution under contingency $\ell$ in \eqref{eq:pfell}.

We elaborate on the modified injections at terminal buses. If bus $m\in\mcP$ is a PQ bus, then $\cbdelta_\ell(\bv_t)$ modifies both of its specifications. If bus $m\in\mcV$ is a PV bus, $\cbdelta_\ell(\bv_t)$ modifies only its active power injection, leaving the voltage specification intact. If bus $m$ is the reference, no modification is made. Bus $n$ is handled similarly. 

The iterative scheme in \eqref{eq:pfell2t} can be equivalently expressed in terms of the basecase AC-PF mapping as
\begin{equation}\label{eq:pf0t}
\tbv_{t+1}=G\left(\bc+\bdelta_\ell(\tbv_t)\right),
\end{equation}
where $\bc$ is the original specifications for the AC-PF problem under contingency $\ell$ in \eqref{eq:pfell}, and the vector $\bdelta_\ell(\tbv_t)$ models how these specifications are modified. Specifically, vector $\bdelta_\ell(\tbv_t)$ collects the real part of the entries of $\cbdelta_\ell(\bv_t)$ corresponding to PQ and PV buses; the imaginary part of the entries of $\cbdelta_\ell(\bv_t)$ corresponding to PQ buses; and contains zeros at the entries corresponding to voltage magnitude specifications. Depending on the types of terminal buses, vector $\bdelta_\ell(\tbv_t)$ may have four nonzero entries (PQ/PQ); three (PQ/PV); two (PQ/reference or PV/PV); or one (PV/reference). We will denote the number of nonzero entries of $\bdelta_\ell(\tbv_t)$ by $N_\ell\in\{1,2,3,4\}$.

\begin{figure}[t]
\centering
\includegraphics[width=0.9\columnwidth]{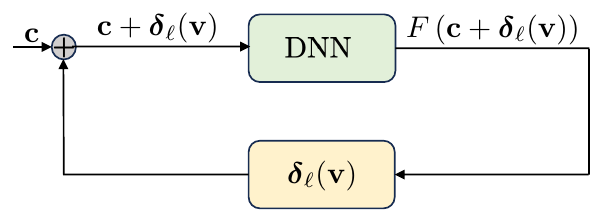}
    \caption{An AC-PF solution under line contingency $\ell$ can be found through a fixed-point iteration including a single DNN in a feedback loop.}
    \label{fig:scheme}
\end{figure}

Inspired by \eqref{eq:pf0t}, we propose predicting AC-PF solutions under line outages by replacing the AC-PF solver with the related DNN. The proposed scheme incorporates  DNN into the next fixed-point iteration:
\begin{equation}\label{eq:Ft}
\tbv_{t+1}=F\left(\bc+\bdelta_\ell(\tbv_t)\right).
\end{equation}
Figure~\ref{fig:scheme} illustrates the scheme. The scheme involves a single DNN, the one trained to solve the basecase AC-PF problem. The advantage of \eqref{eq:Ft} over solving \eqref{eq:pfell} using an AC-PF solver is that the computational complexity can be reduced from $\mcO(N^3)$ to $\mcO(N^2)$ if the DNN depth does not scale with $N$ and the fixed-point iteration in \eqref{eq:Ft} converges fast. We address these two questions numerically through the tests of Section~\ref{sec:tests}.

Before providing numerical evidence, we develop a computational framework to verify whether \eqref{eq:Ft} is convergent. The proposed framework can certify that the iterative scheme converges. It can also determine its region of convergence and the number of iterations required for convergence. 


\section{Convergence Certification of the Fixed-Point Iteration}
\label{sec:certify}
Suppose a system operator has already trained a DNN to predict AC-PF solutions under the basecase topology. Since the DNN parameters $\bw$ are fixed after training, we denote the resulting mapping simply by $\tbv=F(\bc)$. To predict the AC-PF solution under contingency $\ell$, the operator applies the fixed-point iteration shown in \eqref{eq:Ft}. A natural question is whether this iteration converges to a fixed point. This section develops a computational framework for addressing this question. Given the power system model and a trained DNN, the proposed framework can be applied offline to certify the convergence of \eqref{eq:Ft} for each contingency $\ell$. Once convergence has been certified for a particular contingency, the operator can confidently use \eqref{eq:Ft} to predict post-contingency AC-PF solutions.

Although our analysis ignores shunt admittances by setting $y_\ell^s=0$ in \eqref{eq:Y0_construction} for all $\ell$ for simplicity, the tests of Section~\ref{sec:tests} involve the detailed line models. According to the Banach fixed-point theorem, the iterative scheme in \eqref{eq:Ft} converges to a fixed point if the mapping
\begin{equation}\label{eq:H}
H_\ell(\tbv;\bc):=F\left(\bc+\bdelta_\ell(\tbv)\right)  
\end{equation}
satisfies the following two properties across a set $\mcR$ of voltage vectors and a set $\mcC$ of specification vectors $\bc$:
\begin{itemize}
    \item \emph{Self mapping:} If $\tbv\in \mcR$, then $H_\ell(\tbv;\bc)\in\mcR$ for all $\bc\in\mcC$; and
    \item \emph{Contraction:} There exists $L_{H_\ell}<1$ so that 
    \[\|H_\ell(\tbv_1;\bc)-H_\ell(\tbv_2;\bc)\|\leq L_{H_\ell} \|\tbv_1-\tbv_2\|\]
    for all $\tbv_1,\tbv_2\in\mcR$ and all $\bc\in\mcC$.
\end{itemize}

We define sets $\mcR$ and $\mcC$ as the Euclidean balls:
\begin{subequations}\label{eq:balls}
\begin{align}
\mcR & :=\left\{\tbv:\|\tbv-\tbv_0\| \leq R\right\},\label{eq:balls:S}\\
\mcC & :=\left\{\bc:\|\bc-\bc_0\|\leq C\right\}.\label{eq:balls:C}
\end{align}    
\end{subequations}
The voltage center is set to be the so-termed \emph{flat voltage profile} $\tbv_0:=[\bone^\top~\bzero^\top]^\top$, corresponding to an unloaded power system with zero power injections $(\bs_0=\bzero)$. The specification center $\bc_0$ corresponds to this unloaded power system scenario, so that all power injection specifications in $\bc_0$ are set to zero, and all voltage magnitude specifications are set to one. We will assume that the DNN has been trained to perfectly predict the flat voltage profile, i.e., $F(\bc_0)=\tbv_0$.

We first study the self-mapping property of the mapping $H_\ell(\tbv;\bc)$. This property is established in two steps. In the first step, we show that if $\bc\in\mcC$ and $\tbv\in\mcR$, the contingency-corrected specification vector $\bc+\bdelta_\ell(\tbv)$ belongs to a Euclidean ball $\mcC'$ with the same center as $\mcC$ but a larger radius, so that $\mcC\subset \mcC'$. This is formalized in Lemma~\ref{le:Cprime}; all proofs are in the appendix.

\begin{lemma}\label{le:Cprime}
If $\bc\in\mcC$ and $\tbv\in\mcR$, the modified specification vector lies in the Euclidean ball
\begin{equation}\label{eq:Cprime}
\bc+\bdelta_\ell(\tbv)\in\mcC':=\{\bc:\|\bc-\bc_0\|\leq C'\}
\end{equation}
where $C'=C+2\sqrt{2}\bar{y}R(1+R)$ and $\bar{y}:=\max_{\ell}|y_\ell|$.
\end{lemma}

Lemma~\ref{le:Cprime} is practically important because it implies that for the iterations in \eqref{eq:Ft} to work properly, the DNN has to be trained over a larger set of specification scenarios $\mcC'$. Building on Lemma~\ref{le:Cprime}, the next lemma establishes the self-mapping property of $H_\ell(\tbv;\bc)$. 

\begin{lemma}\label{le:selfmap}
If the DNN Lipschitz constant $L_G$ satisfies 
\begin{equation}\label{eq:selfmap}
L_G:=\sup_{\bc\in\mcC'}\|\nabla_\bc F(\bc)\|\leq \frac{R}{C'},
\end{equation}
then $H_\ell(\tbv;\bc)\in\mcR$ for all $\tbv\in\mcR$, $\bc\in\mcC$, and $\ell$. 
\end{lemma}

Lemma~\ref{le:selfmap} establishes that $H_\ell(\tbv;\bc)$ is a self-mapping provided the DNN sensitivity $L_G$ with respect to its input is sufficiently small across $\mcC'$. This result is practically important, as it shows that the DNN must be trained and certified to possess a sufficiently small $L_G$ across the enlarged set $\mcC'$. 

Having established the self-mapping property, we proceed with the contraction property. The Lipschitz constant $L_{H_\ell}$ of the mapping $H_\ell$ can be computed as
\begin{equation}\label{eq:Lh}
L_{H_\ell}=\sup_{\tbv\in\mcR,\bc\in\mcC}\left\| \nabla_{\tbv}H_\ell(\tbv;\bc)\right\|.
\end{equation}
From the chain rule, the Jacobian matrix in \eqref{eq:Lh} can be expressed as the product of two $2N\times 2N$ matrices as
\begin{equation}\label{eq:chain}
\nabla_{\tbv}H_\ell(\tbv;\bc) = \underbrace{\nabla_{\bc+\bdelta_\ell(\tbv)}F\left(\bc+\bdelta_\ell(\tbv)\right)}_{:=\bJ_\ell(\tbv;\bc)}\,
\nabla_{\tbv}\bdelta_\ell(\tbv).
\end{equation}
Recall that vector $\bdelta_\ell(\tbv)$ has only $N_\ell$ nonzero entries with $N_\ell\in\{1,2,3,4\}$, and that those entries depend only on the terminal bus voltages. Therefore, the $2N\times 2N$ Jacobian matrix $\nabla_{\tbv}\bdelta_\ell(\tbv)$ has only $N_\ell$ nonzero rows and 4 nonzero columns. Leveraging this structure, the next lemma provides a factorization of $\nabla_{\tbv}\bdelta_\ell(\tbv)$, which will be helpful for bounding $L_{H_\ell}$.

\begin{lemma}\label{le:Kell}
The Jacobian matrix of the mapping $\bdelta_\ell(\tbv)$ with $\ell=(m,n)$ can be factorized as
\begin{equation}\label{eq:factor}
\nabla_{\tbv}\bdelta_\ell(\tbv)=\bU_\ell \bK_\ell(\tbv) \bV_\ell^\top,
\end{equation}
where $\bV_\ell:=\begin{bmatrix}\be_m & \be_n & \be_{m+N} & \be_{n+N}\end{bmatrix}$; $\bU_{\ell}$ is a $2N\times N_\ell$ matrix formed by keeping $N_\ell$ columns of $\bV_\ell$ depending on the types of buses $m$ and $n$; and $\bK_\ell(\tbv)$ is the $N_\ell\times 4$ Jacobian of the nonzero entries of $\bdelta_\ell(\tbv)$ with respect to terminal bus voltages satisfying
\[\|\bK_{\ell}(\tbv)\|\leq 4|y_\ell|,\quad \forall \tbv\in\mcR.\]
\end{lemma}

Plugging \eqref{eq:chain} and \eqref{eq:factor} into \eqref{eq:Lh} provides
\begin{align}
L_{H_\ell}&=\sup_{\tbv\in\mcR,\bc\in\mcC}\left\| \bJ_\ell(\tbv;\bc)\bU_\ell \bK_\ell(\tbv)\bV_\ell^\top\right\|\nonumber\\
&\leq 4|y_\ell|\,\sup_{\tbv\in\mcR,\bc\in\mcC}\left\| \bJ_\ell(\tbv;\bc)\bU_\ell\right\|\label{eq:bound1}
\end{align}
because of the sub-multiplicative norm property and $\|\bK_\ell(\tbv)\bV_\ell^\top\|=\|\bK_\ell(\tbv)\|$ since $\bV_\ell$ is orthonormal. 

According to \eqref{eq:bound1}, the mapping $H_\ell(\bc+\bdelta_\ell(\tbv))$ is a contraction for all $\tbv\in\mcR$ and $\bc\in\mcC$ if 
\begin{equation}\label{eq:contraction}
L_{G_\ell}:=\sup_{\tbv\in\mcR,\bc\in\mcC}\left\| \bJ_\ell(\tbv;\bc)\bU_\ell\right\|\leq \frac{1}{4|y_\ell|}.
\end{equation}

In summary, the mapping $H_\ell(\bc+\bdelta_\ell(\tbv))$ satisfies the assumptions of the fixed-point theorem if it satisfies \eqref{eq:selfmap} and \eqref{eq:contraction}. The bound in \eqref{eq:selfmap} is on the DNN sensitivity to a general input, while the bound in \eqref{eq:contraction} is on the DNN sensitivity to inputs of particular structure depending on contingency $\ell$. Section~\ref{sec:sdp} explains how to certify whether a given DNN satisfies these bounds by solving SDPs. The numerical tests of Section~\ref{sec:tests} demonstrate that the DNN sensitivities are sufficiently small for the benchmark systems, so that system operators can confidently use the proposed iterative scheme in \eqref{eq:Ft} to predict post-contingency system states.

\section{Bounding DNN Sensitivities via SDPs}\label{sec:sdp}
We present two SDP formulations to bound the Lipschitz constants $L_G$ and $L_{G_\ell}$ for a DNN already trained to predict basecase AC-PF solutions. Computing the exact Lipschitz constant of a DNN is known to be NP-hard~\cite{virmaux2018lipschitz}. We adopt recent works that propose computing upper bounds of DNN Lipschitz constants~\cite{hashemi2021certifying,xu2024eclipse}. Consider the DNN described in \eqref{eq:dnn}. Its Lipschitz constant is upper bounded as $L_G\leq \sqrt{\rho}$, where $\rho$ is the optimal value of the ensuing SDP~\cite{hashemi2021certifying}:
\begin{align}
\min ~&~\rho\label{eq:sdp}\\
    \text{over}~~&~\rho\geq 0, \text{diagonal}~\bLambda_k\succeq 0,~k=0,\ldots,K-1\nonumber\\
    \text{s.to}~~
    & \bM(\bQ_0,\ldots,\bQ_K)\preceq 0.\nonumber
\end{align}    
Matrix $\bM$ is defined as
\begin{align*}
    \bM&:=\sum_{k=0}^{K-1}
    \begin{bmatrix}
        \bW_k\bE_k\\
        \bE_{k+1}
    \end{bmatrix}^{\top}
    \bQ_k
    \begin{bmatrix}
        \bW_k\bE_k\\
        \bE_{k+1}
    \end{bmatrix}\\
    &\quad \quad-
    \begin{bmatrix}
        \bE_0\\
        \bW_K\bE_K
    \end{bmatrix}^{\top}
    \bQ_K
    \begin{bmatrix}
        \bE_0\\
        \bW_K\bE_K
    \end{bmatrix},
\end{align*}
and each $\bE_k$ is a selection matrix so that $\bx_k=\bE_k\bx$ for $k=0,\ldots,K$, if $\bx^\top=\begin{bmatrix}\bx_0^\top& \cdots & \bx_K^\top\end{bmatrix}$. For $k=0,\ldots,K-1$, the matrix $\bQ_k$ is defined as
\begin{equation*}
    \bQ_k:=
    \begin{bmatrix}
        -\bLambda_k\operatorname{dg}(\balpha_k\odot\bbeta_k)
        &
        \frac{1}{2}\bLambda_k\operatorname{dg}(\balpha_k+\bbeta_k)\\
        \frac{1}{2}\bLambda_k\operatorname{dg}(\balpha_k+\bbeta_k)
        &
        -\bLambda_k
    \end{bmatrix},
\end{equation*}
where vectors $\balpha_k$ and $\bbeta_k$ carry lower and upper slope bounds of the activation $\bsigma_k$. Matrix $\bQ_K$ is defined as
\begin{equation*}
    \bQ_K:=
    \begin{bmatrix}
        \rho \bI & 0\\
        0 & -\bI
    \end{bmatrix}.
\end{equation*}

The SDP depends on the DNN weight matrices $\{\bW_k\}_{k=0}^K$, but not the bias vectors $\{\bb_k\}_{k=0}^K$. The Lipschitz bound obtained holds for all inputs $\bc\in\mathbb{R}^{2N}$. If the DNN has ReLUs as activations, the slope bounds are set as $\balpha_k=\bzero$ and $\bbeta_k=\bone$ for all $k$. To improve the conservatism of the bound, subsequent works suggested limiting the ReLU slope bounds if the input lies in a compact set~\cite{hashemi2021certifying,xu2024eclipse}. We briefly summarize the result. If the $i$-th entry in the pre-activation vector $\bz_k$ is bounded as $z_k^i\in[l_k^i,u_k^i]$, the slope bounds for the associated ReLU are
\begin{equation}\label{eq:alpha_beta_three_cases}
(\alpha_k^i,\beta_k^i)
    =
    \begin{cases}
        (0,0), & u_k^i\le 0,\\
        (1,1), & l_k^i\ge 0,\\
        \left(0,\dfrac{u_k^i}{u_k^i-l_k^i}\right),
        & l_k^i<0<u_k^i.
    \end{cases}
\end{equation}
Using tighter slope bounds as $\{(\balpha_k,\bbeta_k)\}_{k=0}^K$ in \eqref{eq:sdp} has been observed to yield significantly tighter bounds on DNN Lipschitz constants~\cite{hashemi2021certifying,xu2024eclipse}. Algorithm~1 of \cite{hashemi2021certifying} provides a computationally-efficient way for computing lower and upper bounds on the entries of $\bz_k$, recursively from the input to the output layer. In our setting, the DNN input lies in the Euclidean ball $C'$. If $\bw_0^i$ is the $i$-th row of $\bW_0$ and $b_0^i$ the $i$-th entry of $\bb_0$, the upper/lower bounds on the $i$-th entry of $\bz_0=\bW_0\bc+\bb_0$ can be found in closed form as $\bc_0^\top\bw_0^i+b_0^i\pm C'\|\bw_0^i\|$. 

Switching to certifying the contraction property, we would like to compute an upper bound on $L_{G_\ell}$ and evaluate the condition in \eqref{eq:contraction}. Different from $L_G$, the constant $L_{G_\ell}$ depends on the partial sensitivity matrix $\nabla_\bc F(\bc)\bU_\ell$, which has only the $N_\ell$ columns of $\nabla_\bc F(\bc)$ associated with contingency $\ell$. Interestingly, constant $L_{G_\ell}$ can be bounded via a neat modification of the SDP. 

We will first modify \eqref{eq:sdp} to compute $L_{G_\ell}$ over $\bc\in\mathbb{R}^{2N}$ for generic slope bounds $\balpha_k=\bzero$ and $\bbeta=\bone$. To this end, partition the contribution of the $N_\ell$ relevant entries of $\bc$ on $\bz_0$ and that of the remaining entries:
\begin{equation}
\bz_0=\bW_0\bP_\ell\bc + \bW_0\bP_\ell^\perp\bc +\bb_0,
\end{equation}
where $\bP_\ell:=\bU_\ell\bU_\ell^\top$ and $\bP_\ell^\perp:=\bI_{2N}-\bP_\ell$. To compute partial sensitivities, we vary $\bU_\ell^\top\bc$ and treat $\bW_0\bP_\ell^\perp\bc +\bb_0$ as a modified bias vector. The original SDP in \eqref{eq:sdp} does not depend on biases. Thus, we can upper bound $\sup_{\bc\in\mathbb{R}^{2N}}\|\nabla_{\bc}F(\bc)\bU_\ell\|$ using a modified SDP where $\bW_0$ has been replaced by $\bW_0\bU_\ell$. As with total sensitivities, this bound may be loose. 

To tighten the bound on $L_{G_\ell}$, we tighten the activation slope bounds in the modified SDP. One option would be to compute upper/lower bounds for each entry of $\bz_0$ over $\bc\in\mcC'$, as we did with total sensitivities. This approach does not utilize the structure $\bc+\bdelta_\ell(\tbv)$ of the DNN input with $\bc\in\mcC$ and $\tbv\in\mcR$. In the proof of Lemma~\ref{le:selfmap}, we showed that if $\tbv\in\mcR$, then $\|\bdelta_\ell(\tbv)\|\leq C_\ell:=2\sqrt{2}|y_\ell|R(1-R)$. Then, the $i$-th entry of $\bz_0$ can be upper bounded by the optimal cost
\begin{align}\label{eq:maxzo}
u_0^i=\max_{\bc,\bdelta_\ell}~&~(\bc+\bdelta_\ell)^\top \bw_0^i  + b_0^i\\
\text{s.to}~&~\|\bc-\bc_0\|\leq C~~\text{and}~~ \|\bdelta_\ell\|\leq C_\ell\nonumber
\end{align}
for $i=1,\ldots,2N$. The problem is separable over $\bc$ and $\bdelta_\ell$. The optimal cost can be found to be
\[u_0^i=\bc_0^\top\bw_0^i+b_0^i + (C+C_\ell)\|\bw_0^i\|.\]
By swapping maximization with minimization, the lower bound on $z_0^i$ can be found as $l_0^i=\bc_0^\top\bw_0^i+b_0^i- (C+C_\ell)\|\bw_0^i\|$.

In a nutshell, to compute an upper bound on $L_G$, we use the SDP in \eqref{eq:sdp} with tightened slope bounds for $\bc\in\mcC'$. To compute an upper bound on $L_{G_\ell}$, we solve \eqref{eq:sdp} for each $\ell$ with $\bW_0$ replaced by $\bW_0\bU_\ell$  and tightened activation slope bounds over $\bc\in\mcC$ and $\|\bdelta_\ell\|\leq C_\ell$. Once the upper bounds on $L_G$ and $L_{G_\ell}$ have been computed, we can certify whether the DNN satisfies conditions \eqref{eq:selfmap} and \eqref{eq:contraction}.

\section{Numerical Tests}
\label{sec:tests}
We evaluated the proposed AC contingency analysis on the IEEE 118-bus system. To assess its scalability, we also evaluated
the proposed method on the larger 6,717-bus Texas7k system,
with results reported in ~\cite{FPM2026}. Among its 186 lines, we considered the 177 non-critical lines as candidate contingencies. We trained a generic fully connected DNN with two hidden layers, each containing 118 neurons, resulting in a number of trainable parameters that scales as $N^2$. We generated 800 training and 200 testing basecase scenarios by perturbing the nominal specifications with zero-mean Gaussian noise. For PQ specifications, the noise standard deviation was $20\%$ during training and $5\%$ during testing. For voltage specifications, it was $0.1$~pu and $0.05$~pu, respectively. The DNN was trained using Adam with a rate of $4\times 10^{-5}$, and training was terminated when the gradient norm fell below $10^{-3}$.

The first test assessed whether the fixed-point iterations satisfy the convergence conditions. Given the trained DNN, we computed the constant $L_G$ using the SDP in \eqref{eq:sdp} and the constants $L_{H_\ell}=4|y_\ell|L_{G_\ell}$ for all $\ell$ using the modified SDP. We set the voltage set radius to $\mcR=0.5$, and determined the specification set radius from the testing set used to train the DNN. We observed that $L_G\leq L_{H_\ell}$ for all $\ell$, so the subsequent analysis relies solely on $L_{H_\ell}$. With the exception of three lines, the contraction condition $L_{H_\ell} \leq 1$ was satisfied. Nonetheless, the fixed-point iteration numerically converged even for those three lines. This is because the contraction condition is only sufficient, and the SDP bounds can be conservative. 

\begin{figure}[t]
\centering
\includegraphics[width=1\linewidth]{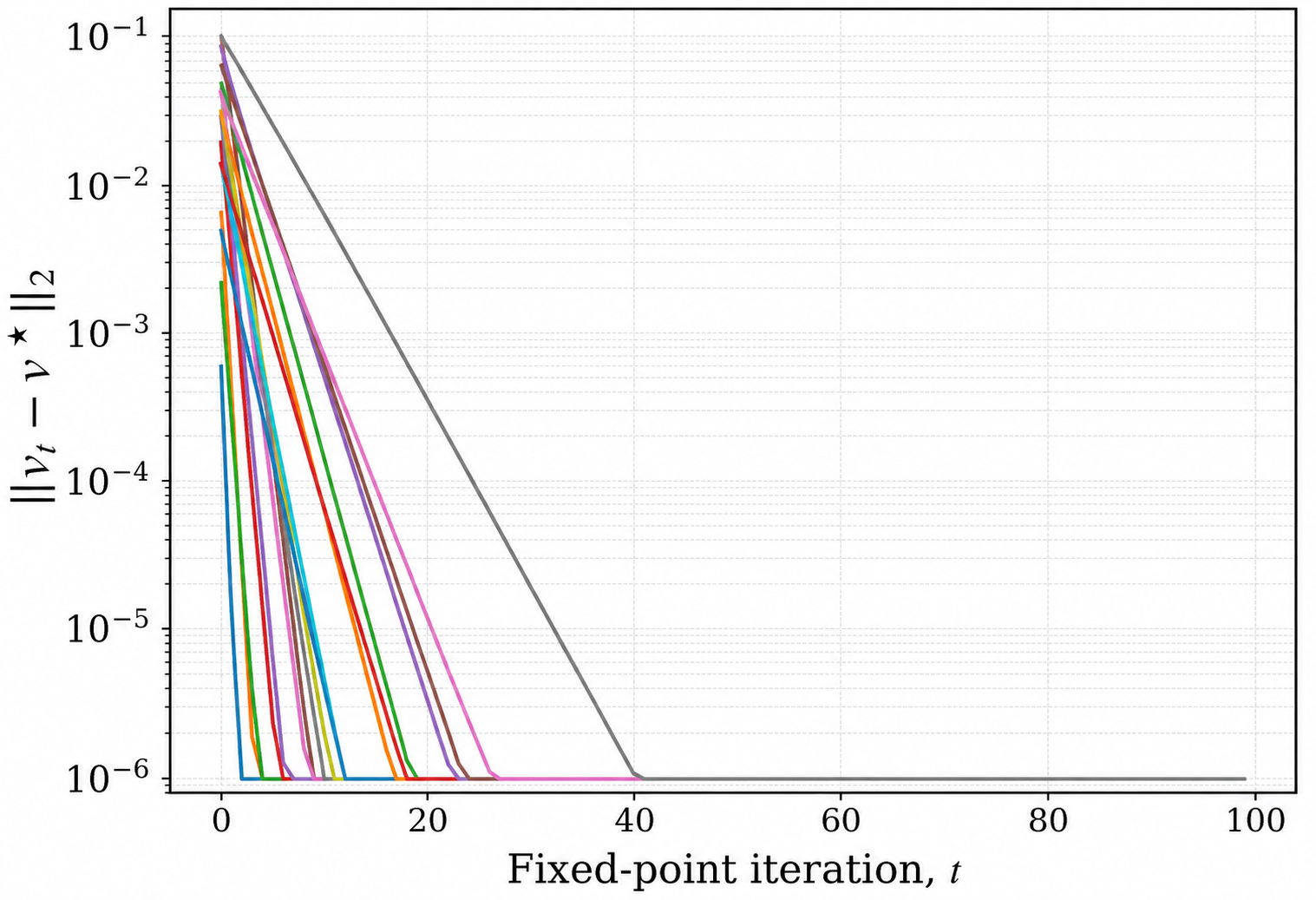}
\caption{Mean Euclidean distance between the fixed-point voltage estimates and their equilibrium for selected contingencies across scenarios. For all lines, the method converges within 50 iterations.}
\label{fig:convergence_selected}
\end{figure}

\begin{figure}[t]
\centering
\includegraphics[width=1\linewidth]{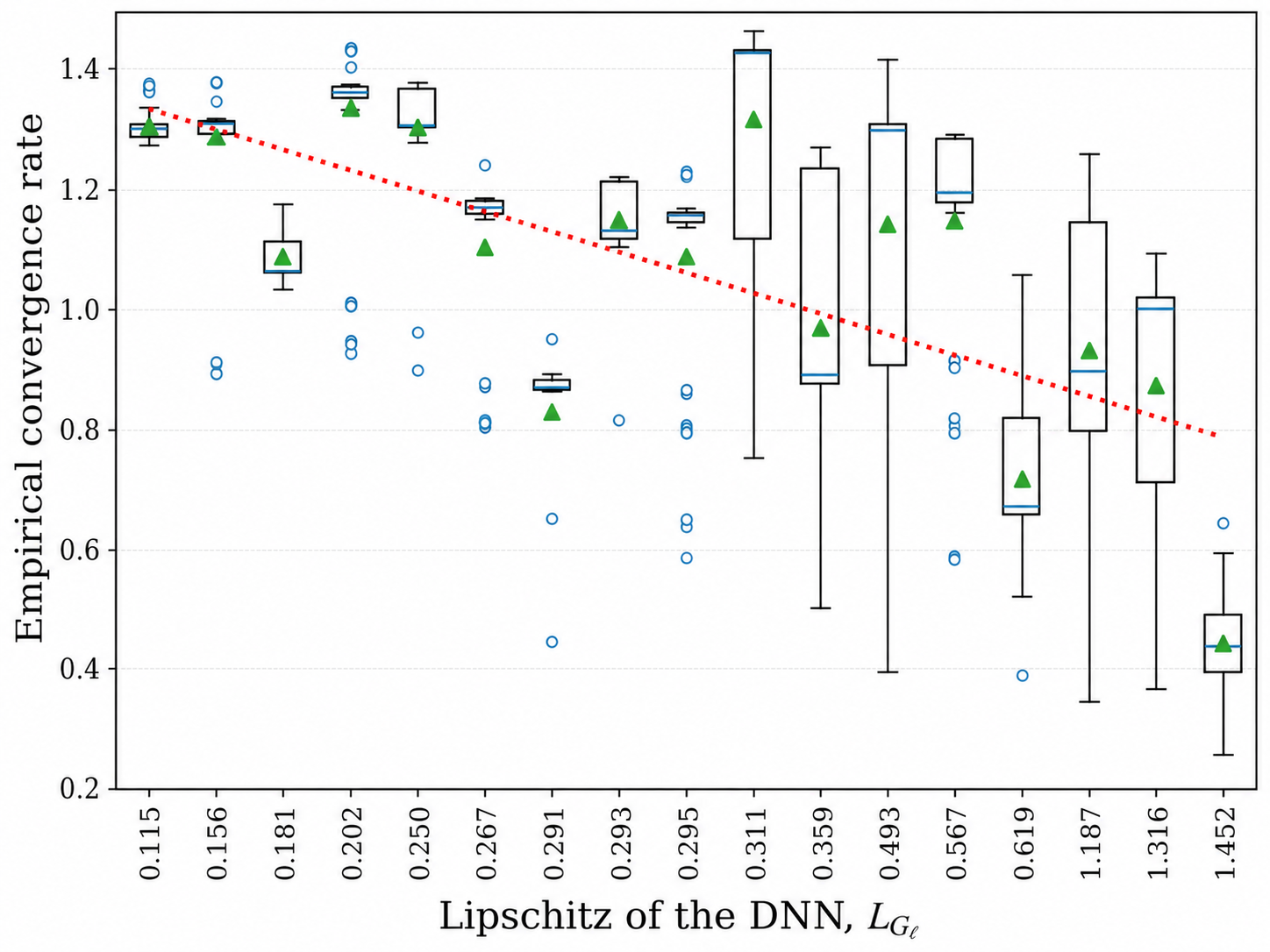}
\caption{Actual convergence rate of the iterative scheme as a function of the Lipschitz constant $L_{H_\ell}$.}
\label{fig:lipschitz_admittance}
\end{figure}

The second test examined the relationship between the empirical convergence rate and the Lipschitz constant $L_{H_\ell}$ across contingencies. Figure~\ref{fig:convergence_selected} shows the Euclidean distance between $\tbv_t$ and the equilibrium point for 17 selected contingencies at ranks $\{1,10,20,\ldots,160,177\}$. The method converges within 50 iterations for all contingencies. Figure~\ref{fig:lipschitz_admittance} plots the mean empirical convergence rate across scenarios for the same contingencies as a function of $L_{H_\ell}$. The convergence rate decreases with the Lipschitz constant, consistent with the contraction analysis in \eqref{eq:contraction}. Variations in the empirical rates arise because the slope is estimated over a finite number of iterations and depends on the initial distance from equilibrium.

\begin{figure}[t]
\centering
\includegraphics[width=1\linewidth]{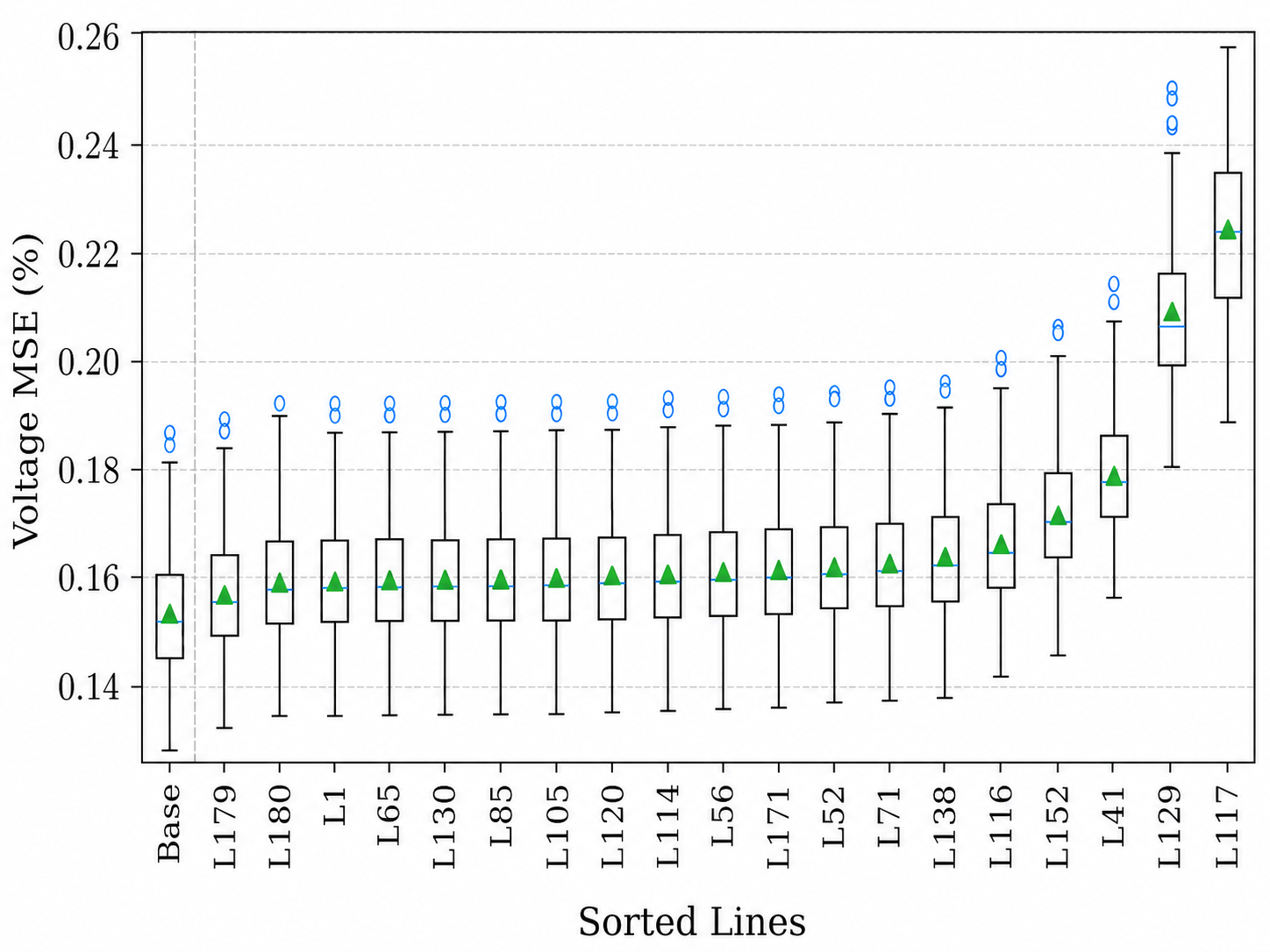}
\vspace*{-2em}
\caption{NMSE between actual and predicted voltages for 17 contingencies.}
\label{fig:voltage_mse_selected}
\end{figure}

\begin{figure}[t]
\centering
\includegraphics[width=1\linewidth]{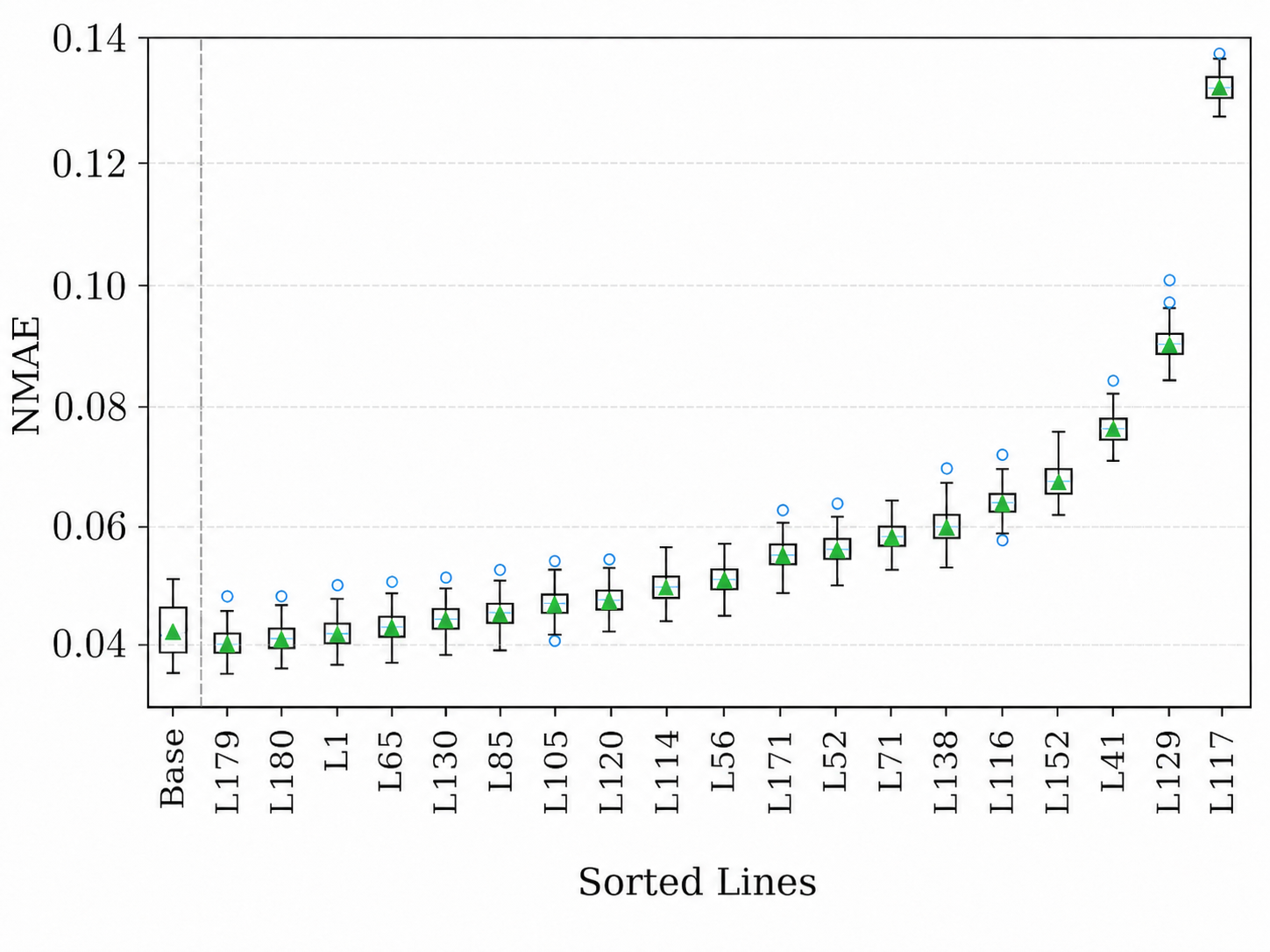}
\vspace*{-2em}
\caption{NMAE between actual and predicted PF specifications across 17 contingencies.}
\label{fig:NMAE}
\end{figure}

The third test assessed the accuracy of the predicted post-contingency states by evaluating the normalized mean squared error (NMSE) between the actual and predicted voltage vectors, and the normalized mean absolute error (NMAE) between the actual AC-PF specifications and those obtained by substituting the predicted voltage states into the AC-PF equations. Figures~\ref{fig:NMAE} and~\ref{fig:voltage_mse_selected} show box plots of these errors for the same selected line contingencies across the 200 testing scenarios, together with the corresponding errors under the basecase topology. For contingencies with smaller Lipschitz values, both metrics remain close to their basecase levels, with noticeable degradation occurring primarily for lines near the end of the Lipschitz ordering.

\begin{table}[t]
\centering
\caption{Runtime comparison over 17,700 instances.}
\label{tab:runtime_comparison}
\begin{tabular}{lcc}
\hline\hline
\textbf{Method} & \textbf{Total (s)} & \textbf{Per instance (ms)}\\
\hline\hline
Proposed method & 3.9369 & 0.2244\\
NR flat & 9.5709 & 0.5407\\
DC model & 0.0103 & 0.0006\\
NR warm start & 8.7270 & 0.4931\\
Jacobian factorization & 4.6770 & 0.2642\\
\hline\hline
\end{tabular}
\vspace*{-1em}
\end{table}

The final test compared the proposed method's computational cost with conventional contingency analysis approaches, including NR with flat and warm starts, Jacobian factorization~\cite{lo2004newton}, and a DC-based solver. We conducted the experiments on a RunPod terminal equipped with an AMD EPYC 9355 32-Core processor and an NVIDIA RTX PRO 6000 Blackwell GPU. We ran the proposed method on the NVIDIA RTX PRO 6000 Blackwell GPU and all other methods on a CPU, since no GPU implementation is available.


Table~\ref{tab:runtime_comparison} compares the total runtime and mean runtime per instance across all contingencies and testing scenarios. Only the DC-based solver achieves lower runtimes than the proposed method, yet at the expense of linearization inaccuracies. Figure~\ref{fig:convergence_selected} explains the computational advantage of the proposed method. With most contingencies converging in about 30 iterations, our method's total complexity is $O(N^2)$ in these tests, compared with $O(N^3)$ per iteration for dense NR. Sparse NR implementations can scale better, typically as $O(N^{1.7})$--$O(N^{2.5})$, while the proposed DNN is highly parallelizable. Its $O(N^2)$ complexity could be reduced further using topology-aware architectures such as GNNs.

\section{Conclusions and Future Work}
\label{sec:conclusion}
This work has proposed a fixed-point formulation that enables a single DNN, trained solely on basecase AC-PF data, to predict AC-PF solutions for arbitrary single-line contingencies. To guarantee convergence of the fixed-point iterations, we established sufficient self-mapping and contraction conditions and developed SDP formulations to certify that the required DNN sensitivity bounds meet these conditions. Numerical tests on the IEEE 118-bus system validate both the convergence and the prediction accuracy of the proposed method, as the SDP-certified bounds satisfy the self-mapping and contraction conditions for almost all tested cases. The results also show an inverse relationship between the convergence rate and the DNN sensitivity. For many contingencies, the prediction errors remain close to the basecase performance, demonstrating that the basecase DNN can be effectively reused in fixed-point iteration. One interesting direction for future work is to extend the formulation to multi-line contingencies and topology reconfiguration. Another key question is whether the exact AC-PF fixed-point map converges. This would clarify whether the observed convergence behavior comes from the AC-PF formulation or from the mapping of the trained DNN.  

\section{Appendix}
\label{sec:appendix}
\begin{proof}[Proof of Lemma~\ref{le:Cprime}]
If $\bc\in\mcC$ and $\tbv\in\mcR$, we get that
\begin{align*}
\|\bc+\bdelta_\ell(\tbv)-\bc_0\|&\leq\|\bc-\bc_0\|+\|\bdelta_\ell(\tbv)\|\nonumber\\
&\leq C + \max_{\tbv\in\mcR}\|\bdelta_\ell(\tbv)\|.
\end{align*}
It suffices to upper bound $\max_{\tbv\in\mcR}\|\bdelta_\ell(\tbv)\|$ for all $\ell$. If we ignore the shunt admittance in \eqref{eq:bdelta}, we get
\begin{equation}\label{eq:bdelta2}
\cbdelta_\ell(\bv)=\bar{y}_\ell \left(\bar{v}_m^-\bar{v}_n\right)\bv\odot (\be_m-\be_n).  
\end{equation}
Recall that $\bdelta_\ell(\tbv)$ is obtained by taking the real values of $\cbdelta_\ell(\bv)$ related to PQ and PV terminal buses, and the imaginary values of $\cbdelta_\ell(\bv)$ related to PQ terminal buses. Because $\bdelta_\ell(\tbv)$ carries a subset or all the values of $\cbdelta_\ell(\bv)$, we get $\|\bdelta_\ell(\tbv)\|\leq \|\cbdelta_\ell(\bv)\|$. From \eqref{eq:bdelta2}, we obtain
\begin{equation}\label{eq:bound2b}
\|\cbdelta_\ell(\bv)\|=|y_\ell|\cdot |v_m-v_n|\cdot \sqrt{|v_m|^2+|v_n|^2}.    
\end{equation}
If $\tbv\in\mcR$, then $\|\bv-\bone\|\leq R$ and so $\|\bv-\bone\|_{\infty}\leq R$. This implies that $|v_m|\leq (1+R)$, $|v_n|\leq (1+R)$, and $|v_m-v_n|\leq 2R$ for all entries of $\bv$. Hence, the norm in \eqref{eq:bound2b} is upper bounded by $2\sqrt{2}|y_\ell|R(1+R)$ for all $\ell$. The radius $C'$ is a universal bound for all $\ell$.
\end{proof}

\begin{proof}[Proof of Lemma~\ref{le:selfmap}]
The goal is to show that $\|H_\ell(\tbv;\bc)-\tbv_0\|\leq R$ for all $\tbv\in\mcR$ and $\bc\in\mcC$. To this end, we have
\begin{align*}
\|H_\ell(\tbv;\bc)-\tbv_0\|
&= \|F(\bc+\bdelta_\ell(\tbv))-F(\bc_0)\|\\
&\leq L_G \, \|\bc+\bdelta_\ell(\tbv)-\bc_0\|\leq L_G \, C'.
\end{align*}
The equality follows from the definition of $H_\ell$ in \eqref{eq:H} and assuming that the DNN has been trained to predict $\tbv_0$ when presented $\bc_0$. The first inequality follows from the DNN Lipschitz constant definition over set $\mcC'$. The second inequality follows from Lemma~\ref{le:selfmap}. The condition in \eqref{eq:selfmap} ensures that $L_G\, C'\leq R$. 
\end{proof}

\begin{proof}[Proof of Lemma~\ref{le:Kell}]
We first consider the case where both terminal buses are PQ buses. In this case, vector $\bdelta_\ell(\tbv)$ has four nonzero entries, corresponding to the two PQ specifications per terminal bus:
\begin{subequations}\label{eq:tbdelta}
\begin{align}
\delta_{\ell,m}(\tbv) &=\real\left\{\bar{y}_\ell v_m(\bar{v}_m-\bar{v}_n)\right\}\label{eq:tbdelta:a}\\
\delta_{\ell,n}(\tbv) &=-\real\left\{\bar{y}_\ell v_n(\bar{v}_m-\bar{v}_n)\right\}\label{eq:tbdelta:b}\\
\delta_{\ell,m+N}(\tbv) &=\imag\left\{\bar{y}_\ell v_m(\bar{v}_m-\bar{v}_n)\right\}\label{eq:tbdelta:c}\\
\delta_{\ell,n+N}(\tbv) &=-\imag\left\{\bar{y}_\ell v_n(\bar{v}_m-\bar{v}_n)\right\}\label{eq:tbdelta:d}.
\end{align}    
\end{subequations}
The four functions are quadratic functions of only four entries of $\tbv$, which can be collected in vector $\tbv_\ell:=[\tilde{v}_m~\tilde{v}_n~\tilde{v}_{m+N}~\tilde{v}_{n+N}]^\top$. 

Given its sparsity pattern, matrix $\nabla_{\tbv}\bdelta_\ell(\tbv)$ can be factorized as in \eqref{eq:factor} with $\bU_\ell=\bV_\ell$ and $\bK_{\ell}(\tbv)$ being the Jacobian of the four quadratic functions. Since the functions are quadratic, the Jacobian is linear in $\tbv_\ell$:
\begin{equation}\label{eq:Kell}
    \bK_\ell(\tbv)=\sum_{i=1}^{4}\tilde{v}_{\ell,i} \bA_{\ell}^i,
\end{equation}
where matrices $\bA_{\ell}^i$ depend on $y_\ell=g_\ell+jb_\ell$ as follows
\begin{align*}
\bA_{\ell}^1&=
\begin{bmatrix}
2g_\ell & -g_\ell & 0 & b_\ell\\
0 & -g_\ell & 0 & -b_\ell\\
-2b_\ell & b_\ell & 0 & g_\ell\\
0 & b_\ell & 0 & -g_\ell
\end{bmatrix},\\
\bA_{\ell}^2&=
\begin{bmatrix}
-g_\ell & 0 & -b_\ell & 0\\
-g_\ell & 2g_\ell & b_\ell & 0\\
b_\ell & 0 & -g_\ell & 0\\
b_\ell & -2b_\ell & g_\ell & 0
\end{bmatrix},\\
\bA_{\ell}^3&=
\begin{bmatrix}
0 & -b_\ell & 2g_\ell & -g_\ell\\
0 & b_\ell & 0 & -g_\ell\\
0 & -g_\ell & -2b_\ell & b_\ell\\
0 & g_\ell & 0 & b_\ell
\end{bmatrix},\\
\bA_{\ell}^4&=
\begin{bmatrix}
b_\ell & 0 & -g_\ell & 0\\
-b_\ell & 0 & -g_\ell & 2g_\ell\\
g_\ell & 0 & b_\ell & 0\\
-g_\ell & 0 & b_\ell & -2b_\ell
\end{bmatrix}.
\end{align*}
Matrix $\bK_{\ell}(\tbv)$ can be alternatively expressed as
\begin{equation}\label{eq:Kell2}
\bK_\ell(\tbv)= \bA_\ell \left(\tbv_\ell\otimes \bI_4\right),
\end{equation}
where $\bA_\ell:=\begin{bmatrix}\bA_{\ell}^1 & \bA_{\ell}^2 & \bA_{\ell}^3 & \bA_{\ell}^4\end{bmatrix}$ and $\otimes$ is the Kronecker product. From the sub-multiplicative property of the spectral norm, it follows that $\|\bK_\ell(\tbv)\|\leq \|\bA_{\ell}\| \cdot \|\tbv_\ell\otimes \bI_4\|=\|\bA_{\ell}\| \cdot \|\tbv_\ell\|$.

Since $\tbv\in \mcR$, we get $\|\bv-\bone\|=\|\tbv-\tbv_0\|\leq R$ and $\|\bv-\bone\|_{\infty}\leq \|\bv-\bone\|\leq R$ so that $|v_m|\leq 1+R$ and $|v_n|\leq 1+R$. It hence follows that $\|\tbv_\ell\|=\|\bv_\ell\|\leq \sqrt{2}\|\bv_\ell\|_{\infty}\leq \sqrt{2}(1+R)$. Regarding $\bA_\ell$, after some algebraic manipulations, it can be shown that the largest eigenvalue of $\bA_\ell\bA_\ell^\top$ is $16|y_\ell|^2$, so that $\|\bA_\ell\|=4|y_\ell|$.

Other combinations of terminal bus types can be handled by selecting a subset of rows from $\bK_\ell(\tbv)$ and by forming $\bU_\ell$ by selecting the corresponding columns from $\bV_\ell$. For example, if $m\in\mcP$ and $n\in\mcV$, the $(n+N)$-th entry of $\bdelta_\ell(\tbv)$ is zero, and so the quadratic function in \eqref{eq:tbdelta:d} is ignored. In this case, we have $N_\ell=3$, $\bU_\ell=\begin{bmatrix}\be_m & \be_n & \be_{m+N}\end{bmatrix}$, and $\bK_\ell(\tbv)$ is obtained from \eqref{eq:Kell} by dropping the fourth row from each $\bA_\ell^i$. Because dropping rows from a matrix cannot increase its spectral norm, the bound $\|\bK_\ell(\tbv)\|\leq 4|y_\ell|$ remains valid. In general, if a terminal bus is a PQ bus, it contributes two rows to $\bK_\ell(\tbv)$; if it is a PV bus, it contributes one row; and if it is the reference bus, it contributes no row to $\bK_\ell(\tbv)$.
\end{proof}

\bibliographystyle{IEEEtran}
\bibliography{myabbrv_1, sample2}

\color{blue}
\color{black}
\section*{Acknowledgments}
This work is supported by US NSF grant 2434502.

\end{document}